\documentclass[conference,10pt,twoside,twocolumn]{IEEEtran}

\usepackage{amsmath,graphicx}

\usepackage[T1]{fontenc}
\usepackage{graphicx}
\usepackage{amssymb}
\usepackage{amsmath}
\usepackage{amsthm}
\usepackage{subfigure}
\usepackage{booktabs} 
\usepackage{multirow}
\usepackage{microtype}
\usepackage{cite}
\usepackage{subfigure}
\usepackage{xcolor} 
\usepackage{url}
\usepackage{balance}
\usepackage{colortbl}

\usepackage{amssymb}
\usepackage{amsfonts}
\usepackage{mathrsfs}
\usepackage{xspace}
\usepackage{bm}
\usepackage{upgreek}

\newcommand{\safemath}[2]{\newcommand{#1}{\ensuremath{#2}\xspace}}

\safemath{\bma}{\mathbf{a}}
\safemath{\bmb}{\mathbf{b}}
\safemath{\bmc}{\mathbf{c}}
\safemath{\bmd}{\mathbf{d}}
\safemath{\bme}{\mathbf{e}}
\safemath{\bmf}{\mathbf{f}}
\safemath{\bmg}{\mathbf{g}}
\safemath{\bmh}{\mathbf{h}}
\safemath{\bmi}{\mathbf{i}}
\safemath{\bmj}{\mathbf{j}}
\safemath{\bmk}{\mathbf{k}}
\safemath{\bml}{\mathbf{l}}
\safemath{\bmm}{\mathbf{m}}
\safemath{\bmn}{\mathbf{n}}
\safemath{\bmo}{\mathbf{o}}
\safemath{\bmp}{\mathbf{p}}
\safemath{\bmq}{\mathbf{q}}
\safemath{\bmr}{\mathbf{r}}
\safemath{\bms}{\mathbf{s}}
\safemath{\bmt}{\mathbf{t}}
\safemath{\bmu}{\mathbf{u}}
\safemath{\bmv}{\mathbf{v}}
\safemath{\bmw}{\mathbf{w}}
\safemath{\bmx}{\mathbf{x}}
\safemath{\bmy}{\mathbf{y}}
\safemath{\bmz}{\mathbf{z}}
\safemath{\bmzero}{\mathbf{0}}
\safemath{\bmone}{\mathbf{1}}
\safemath{\Bell}{\ensuremath{\boldsymbol\ell}}

\bmdefine{\biad}{a}
\bmdefine{\bibd}{b}
\bmdefine{\bicd}{c}
\bmdefine{\bidd}{d}
\bmdefine{\bied}{e}
\bmdefine{\bifd}{f}
\bmdefine{\bigd}{g}
\bmdefine{\bihd}{h}
\bmdefine{\biid}{i}
\bmdefine{\bijd}{j}
\bmdefine{\bikd}{k}
\bmdefine{\bild}{l}
\bmdefine{\bimd}{m}
\bmdefine{\bind}{n}
\bmdefine{\biod}{o}
\bmdefine{\bipd}{p}
\bmdefine{\biqd}{q}
\bmdefine{\bird}{r}
\bmdefine{\bisd}{s}
\bmdefine{\bitd}{t}
\bmdefine{\biud}{u}
\bmdefine{\bivd}{v}
\bmdefine{\biwd}{w}
\bmdefine{\bixd}{x}
\bmdefine{\biyd}{y}
\bmdefine{\bizd}{z}

\bmdefine{\bixid}{\xi}
\bmdefine{\bilambdad}{\lambda}
\bmdefine{\bimud}{\mu}
\bmdefine{\bithetad}{\theta}
\bmdefine{\biphid}{\phi}
\bmdefine{\bideltad}{\delta}

\safemath{\bmia}{\biad}
\safemath{\bmib}{\bibd}
\safemath{\bmic}{\bicd}
\safemath{\bmid}{\bidd}
\safemath{\bmie}{\bied}
\safemath{\bmif}{\bifd}
\safemath{\bmig}{\bigd}
\safemath{\bmih}{\bihd}
\safemath{\bmii}{\biid}
\safemath{\bmij}{\bijd}
\safemath{\bmik}{\bikd}
\safemath{\bmil}{\bild}
\safemath{\bmim}{\bimd}
\safemath{\bmin}{\bind}
\safemath{\bmio}{\biod}
\safemath{\bmip}{\bipd}
\safemath{\bmiq}{\biqd}
\safemath{\bmir}{\bird}
\safemath{\bmis}{\bisd}
\safemath{\bmit}{\bitd}
\safemath{\bmiu}{\biud}
\safemath{\bmiv}{\bivd}
\safemath{\bmiw}{\biwd}
\safemath{\bmix}{\bixd}
\safemath{\bmiy}{\biyd}
\safemath{\bmiz}{\bizd}

\safemath{\bmxi}{\bixid}
\safemath{\bmlambda}{\bilambdad}
\safemath{\bmmu}{\bimud}
\safemath{\bmtheta}{\bithetad}
\safemath{\bmphi}{\biphid}
\safemath{\bmdelta}{\bideltad}

\safemath{\bA}{\mathbf{A}}
\safemath{\bB}{\mathbf{B}}
\safemath{\bC}{\mathbf{C}}
\safemath{\bD}{\mathbf{D}}
\safemath{\bE}{\mathbf{E}}
\safemath{\bF}{\mathbf{F}}
\safemath{\bG}{\mathbf{G}}
\safemath{\bH}{\mathbf{H}}
\safemath{\bI}{\mathbf{I}}
\safemath{\bJ}{\mathbf{J}}
\safemath{\bK}{\mathbf{K}}
\safemath{\bL}{\mathbf{L}}
\safemath{\bM}{\mathbf{M}}
\safemath{\bN}{\mathbf{N}}
\safemath{\bO}{\mathbf{O}}
\safemath{\bP}{\mathbf{P}}
\safemath{\bQ}{\mathbf{Q}}
\safemath{\bR}{\mathbf{R}}
\safemath{\bS}{\mathbf{S}}
\safemath{\bT}{\mathbf{T}}
\safemath{\bU}{\mathbf{U}}
\safemath{\bV}{\mathbf{V}}
\safemath{\bW}{\mathbf{W}}
\safemath{\bX}{\mathbf{X}}
\safemath{\bY}{\mathbf{Y}}
\safemath{\bZ}{\mathbf{Z}}

\safemath{\bZero}{\mathbf{0}}
\safemath{\bOne}{\mathbf{1}}
\safemath{\bDelta}{\mathbf{\Delta}}
\safemath{\bLambda}{\mathbf{\UpLambda}}
\safemath{\bPhi}{\mathbf{\Upphi}}
\safemath{\bSigma}{\mathbf{\Upsigma}}
\safemath{\bOmega}{\mathbf{\Upomega}}
\safemath{\bTheta}{\mathbf{\Uptheta}}

\bmdefine{\biAd}{A}
\bmdefine{\biBd}{B}
\bmdefine{\biCd}{C}
\bmdefine{\biDd}{D}
\bmdefine{\biEd}{E}
\bmdefine{\biFd}{F}
\bmdefine{\biGd}{G}
\bmdefine{\biHd}{H}
\bmdefine{\biId}{I}
\bmdefine{\biJd}{J}
\bmdefine{\biKd}{K}
\bmdefine{\biLd}{L}
\bmdefine{\biMd}{M}
\bmdefine{\biOd}{N}
\bmdefine{\biPd}{O}
\bmdefine{\biQd}{P}
\bmdefine{\biRd}{R}
\bmdefine{\biSd}{S}
\bmdefine{\biTd}{T}
\bmdefine{\biUd}{U}
\bmdefine{\biVd}{V}
\bmdefine{\biWd}{W}
\bmdefine{\biXd}{X}
\bmdefine{\biYd}{Y}
\bmdefine{\biZd}{Z}

\bmdefine{\biDelta}{\Delta}
\bmdefine{\biLambda}{\Lambda}
\bmdefine{\biPhi}{\Phi}
\bmdefine{\biSigma}{\Sigma}
\bmdefine{\biOmega}{\Omega}
\bmdefine{\biTheta}{\Theta}

\safemath{\bimA}{\biAd}
\safemath{\bimB}{\biBd}
\safemath{\bimC}{\biCd}
\safemath{\bimD}{\biDd}
\safemath{\bimE}{\biEd}
\safemath{\bimF}{\biFd}
\safemath{\bimG}{\biGd}
\safemath{\bimH}{\biHd}
\safemath{\bimI}{\biId}
\safemath{\bimJ}{\biJd}
\safemath{\bimK}{\biKd}
\safemath{\bimL}{\biLd}
\safemath{\bimM}{\biMd}
\safemath{\bimN}{\biNd}
\safemath{\bimO}{\biOd}
\safemath{\bimP}{\biPd}
\safemath{\bimQ}{\biQd}
\safemath{\bimR}{\biRd}
\safemath{\bimS}{\biSd}
\safemath{\bimT}{\biTd}
\safemath{\bimU}{\biUd}
\safemath{\bimV}{\biVd}
\safemath{\bimW}{\biWd}
\safemath{\bimX}{\biXd}
\safemath{\bimY}{\biYd}
\safemath{\bimZ}{\biZd}

\safemath{\bimDelta}{\biDelta}
\safemath{\bimLambda}{\biLambda}
\safemath{\bimPhi}{\biPhi}
\safemath{\bimSigma}{\biSigma}
\safemath{\bimOmega}{\biOmega}
\safemath{\bimTheta}{\biTheta}

\safemath{\setA}{\mathcal{A}}
\safemath{\setB}{\mathcal{B}}
\safemath{\setC}{\mathcal{C}}
\safemath{\setD}{\mathcal{D}}
\safemath{\setE}{\mathcal{E}}
\safemath{\setF}{\mathcal{F}}
\safemath{\setG}{\mathcal{G}}
\safemath{\setH}{\mathcal{H}}
\safemath{\setI}{\mathcal{I}}
\safemath{\setJ}{\mathcal{J}}
\safemath{\setK}{\mathcal{K}}
\safemath{\setL}{\mathcal{L}}
\safemath{\setM}{\mathcal{M}}
\safemath{\setN}{\mathcal{N}}
\safemath{\setO}{\mathcal{O}}
\safemath{\setP}{\mathcal{P}}
\safemath{\setQ}{\mathcal{Q}}
\safemath{\setR}{\mathcal{R}}
\safemath{\setS}{\mathcal{S}}
\safemath{\setT}{\mathcal{T}}
\safemath{\setU}{\mathcal{U}}
\safemath{\setV}{\mathcal{V}}
\safemath{\setW}{\mathcal{W}}
\safemath{\setX}{\mathcal{X}}
\safemath{\setY}{\mathcal{Y}}
\safemath{\setZ}{\mathcal{Z}}
\safemath{\emptySet}{\varnothing}

\safemath{\colA}{\mathscr{A}}
\safemath{\colB}{\mathscr{B}}
\safemath{\colC}{\mathscr{C}}
\safemath{\colD}{\mathscr{D}}
\safemath{\colE}{\mathscr{E}}
\safemath{\colF}{\mathscr{F}}
\safemath{\colG}{\mathscr{G}}
\safemath{\colH}{\mathscr{H}}
\safemath{\colI}{\mathscr{I}}
\safemath{\colJ}{\mathscr{J}}
\safemath{\colK}{\mathscr{K}}
\safemath{\colL}{\mathscr{L}}
\safemath{\colM}{\mathscr{M}}
\safemath{\colN}{\mathscr{N}}
\safemath{\colO}{\mathscr{O}}
\safemath{\colP}{\mathscr{P}}
\safemath{\colQ}{\mathscr{Q}}
\safemath{\colR}{\mathscr{R}}
\safemath{\colS}{\mathscr{S}}
\safemath{\colT}{\mathscr{T}}
\safemath{\colU}{\mathscr{U}}
\safemath{\colV}{\mathscr{V}}
\safemath{\colW}{\mathscr{W}}
\safemath{\colX}{\mathscr{X}}
\safemath{\colY}{\mathscr{Y}}
\safemath{\colZ}{\mathscr{Z}}

\safemath{\opA}{\mathbb{A}}
\safemath{\opB}{\mathbb{B}}
\safemath{\opC}{\mathbb{C}}
\safemath{\opD}{\mathbb{D}}
\safemath{\opE}{\mathbb{E}}
\safemath{\opF}{\mathbb{F}}
\safemath{\opG}{\mathbb{G}}
\safemath{\opH}{\mathbb{H}}
\safemath{\opI}{\mathbb{I}}
\safemath{\opJ}{\mathbb{J}}
\safemath{\opK}{\mathbb{K}}
\safemath{\opL}{\mathbb{L}}
\safemath{\opM}{\mathbb{M}}
\safemath{\opN}{\mathbb{N}}
\safemath{\opO}{\mathbb{O}}
\safemath{\opP}{\mathbb{P}}
\safemath{\opQ}{\mathbb{Q}}
\safemath{\opR}{\mathbb{R}}
\safemath{\opS}{\mathbb{S}}
\safemath{\opT}{\mathbb{T}}
\safemath{\opU}{\mathbb{U}}
\safemath{\opV}{\mathbb{V}}
\safemath{\opW}{\mathbb{W}}
\safemath{\opX}{\mathbb{X}}
\safemath{\opY}{\mathbb{Y}}
\safemath{\opZ}{\mathbb{Z}}
\safemath{\opZero}{\mathbb{O}}
\safemath{\identityop}{\opI}

\safemath{\veca}{\bma}
\safemath{\vecb}{\bmb}
\safemath{\vecc}{\bmc}
\safemath{\vecd}{\bmd}
\safemath{\vece}{\bme}
\safemath{\vecf}{\bmf}
\safemath{\vecg}{\bmg}
\safemath{\vech}{\bmh}
\safemath{\veci}{\bmi}
\safemath{\vecj}{\bmj}
\safemath{\veck}{\bmk}
\safemath{\vecl}{\bml}
\safemath{\vecm}{\bmm}
\safemath{\vecn}{\bmn}
\safemath{\veco}{\bmo}
\safemath{\vecp}{\bmp}
\safemath{\vecq}{\bmq}
\safemath{\vecr}{\bmr}
\safemath{\vecs}{\bms}
\safemath{\vect}{\bmt}
\safemath{\vecu}{\bmu}
\safemath{\vecv}{\bmv}
\safemath{\vecw}{\bmw}
\safemath{\vecx}{\bmx}
\safemath{\vecy}{\bmy}
\safemath{\vecz}{\bmz}

\safemath{\veczero}{\bmzero}
\safemath{\vecone}{\bmone}
\safemath{\vecxi}{\bmxi}
\safemath{\veclambda}{\bmlambda}
\safemath{\vecmu}{\bmmu}
\safemath{\vectheta}{\bmtheta}
\safemath{\vecphi}{\bmphi}
\safemath{\vecdelta}{\bmdelta}

\safemath{\matA}{\bA}
\safemath{\matB}{\bB}
\safemath{\matC}{\bC}
\safemath{\matD}{\bD}
\safemath{\matE}{\bE}
\safemath{\matF}{\bF}
\safemath{\matG}{\bG}
\safemath{\matH}{\bH}
\safemath{\matI}{\bI}
\safemath{\matJ}{\bJ}
\safemath{\matK}{\bK}
\safemath{\matL}{\bL}
\safemath{\matM}{\bM}
\safemath{\matN}{\bN}
\safemath{\matO}{\bO}
\safemath{\matP}{\bP}
\safemath{\matQ}{\bQ}
\safemath{\matR}{\bR}
\safemath{\matS}{\bS}
\safemath{\matT}{\bT}
\safemath{\matU}{\bU}
\safemath{\matV}{\bV}
\safemath{\matW}{\bW}
\safemath{\matX}{\bX}
\safemath{\matY}{\bY}
\safemath{\matZ}{\bZ}
\safemath{\matzero}{\bmzero}

\safemath{\matDelta}{\bDelta}
\safemath{\matLambda}{\bLambda}
\safemath{\matPhi}{\bPhi}
\safemath{\matSigma}{\bSigma}
\safemath{\matOmega}{\bOmega}
\safemath{\matTheta}{\bTheta}

\safemath{\matidentity}{\matI}
\safemath{\matone}{\matO}

\safemath{\rnda}{A}
\safemath{\rndb}{B}
\safemath{\rndc}{C}
\safemath{\rndd}{D}
\safemath{\rnde}{E}
\safemath{\rndf}{F}
\safemath{\rndg}{G}
\safemath{\rndh}{H}
\safemath{\rndi}{I}
\safemath{\rndj}{J}
\safemath{\rndk}{K}
\safemath{\rndl}{L}
\safemath{\rndm}{M}
\safemath{\rndn}{N}
\safemath{\rndo}{O}
\safemath{\rndp}{P}
\safemath{\rndq}{Q}
\safemath{\rndr}{R}
\safemath{\rnds}{S}
\safemath{\rndt}{T}
\safemath{\rndu}{U}
\safemath{\rndv}{V}
\safemath{\rndw}{W}
\safemath{\rndx}{X}
\safemath{\rndy}{Y}
\safemath{\rndz}{Z}

\safemath{\rveca}{\bimA}
\safemath{\rvecb}{\bimB}
\safemath{\rvecc}{\bimC}
\safemath{\rvecd}{\bimD}
\safemath{\rvece}{\bimE}
\safemath{\rvecf}{\bimF}
\safemath{\rvecg}{\bimG}
\safemath{\rvech}{\bimH}
\safemath{\rveci}{\bimI}
\safemath{\rvecj}{\bimJ}
\safemath{\rveck}{\bimK}
\safemath{\rvecl}{\bimL}
\safemath{\rvecm}{\bimM}
\safemath{\rvecn}{\bimN}
\safemath{\rveco}{\bomO}
\safemath{\rvecp}{\bimP}
\safemath{\rvecq}{\bimQ}
\safemath{\rvecr}{\bimR}
\safemath{\rvecs}{\bimS}
\safemath{\rvect}{\bimT}
\safemath{\rvecu}{\bimU}
\safemath{\rvecv}{\bimV}
\safemath{\rvecw}{\bimW}
\safemath{\rvecx}{\bimX}
\safemath{\rvecy}{\bimY}
\safemath{\rvecz}{\bimZ}

\safemath{\rvecxi}{\bmxi}
\safemath{\rveclambda}{\bmlambda}
\safemath{\rvecmu}{\bmmu}
\safemath{\rvectheta}{\bmtheta}
\safemath{\rvecphi}{\bmphi}

\safemath{\rmatA}{\bimA}
\safemath{\rmatB}{\bimB}
\safemath{\rmatC}{\bimC}
\safemath{\rmatD}{\bimD}
\safemath{\rmatE}{\bimE}
\safemath{\rmatF}{\bimF}
\safemath{\rmatG}{\bimG}
\safemath{\rmatH}{\bimH}
\safemath{\rmatI}{\bimI}
\safemath{\rmatJ}{\bimJ}
\safemath{\rmatK}{\bimK}
\safemath{\rmatL}{\bimL}
\safemath{\rmatM}{\bimM}
\safemath{\rmatN}{\bimN}
\safemath{\rmatO}{\bimO}
\safemath{\rmatP}{\bimP}
\safemath{\rmatQ}{\bimQ}
\safemath{\rmatR}{\bimR}
\safemath{\rmatS}{\bimS}
\safemath{\rmatT}{\bimT}
\safemath{\rmatU}{\bimU}
\safemath{\rmatV}{\bimV}
\safemath{\rmatW}{\bimW}
\safemath{\rmatX}{\bimX}
\safemath{\rmatY}{\bimY}
\safemath{\rmatZ}{\bimZ}

\safemath{\rmatDelta}{\bimDelta}
\safemath{\rmatLambda}{\bimLambda}
\safemath{\rmatPhi}{\bimPhi}
\safemath{\rmatSigma}{\bimSigma}
\safemath{\rmatOmega}{\bimOmega}
\safemath{\rmatTheta}{\bimTheta}

\usepackage{amssymb}
\usepackage{amsfonts}
\usepackage{mathrsfs}
\usepackage{xspace}
\usepackage{bm}
\usepackage{fancyref}
\usepackage{textcomp}

\usepackage{multirow}
\usepackage{stmaryrd}

\newenvironment{textbmatrix}{	\setlength{\arraycolsep}{2.5pt}%
								\left[\begin{matrix}}{\end{matrix}\right]%
								\raisebox{0.08ex}{\vphantom{M}}}

\def\be{\begin{equation}}
\def\ee{\end{equation}}
\def\een{\nonumber \end{equation}}
\def\mat{\begin{bmatrix}}
\def\emat{\end{bmatrix}}
\def\btm{\begin{textbmatrix}}
\def\etm{\end{textbmatrix}}

\def\ba#1\ea{\begin{align}#1\end{align}}
\def\bas#1\eas{\begin{align*}#1\end{align*}}
\def\bs#1\es{\begin{split}#1\end{split}}
\def\bg#1\eg{\begin{gather}#1\end{gather}}
\def\bml#1\eml{\begin{multline}#1\end{multline}}
\def\bi#1\ei{\begin{itemize}#1\end{itemize}}

\newcommand{\lefto}{\mathopen{}\left}

\DeclareMathOperator{\Exop}{\opE}			
\DeclareMathOperator{\Varop}{\opV\mkern-1mu\mathrm{ar}}
\DeclareMathOperator{\conv}{\star}			

\newcommand{\Ex}[1]{\ensuremath{\Exop\lefto[#1\right]}} 	
\newcommand{\Var}[1]{\ensuremath{\Varop\lefto[#1\right]}} 

\newcommand{\vecnorm}[1]{\lefto\lVert#1\right\rVert}		
\newcommand{\frobnorm}[1]{\vecnorm{#1}_{\text{F}}}	

\safemath{\dirac}{\delta}					
\safemath{\krond}{\dirac}					

\safemath{\upto}{\uparrow}
\safemath{\downto}{\downarrow}
\safemath{\iu}{j}							
\safemath{\ev}{\lambda}						
\safemath{\hilseqspace}{l^{2}}				
\newcommand{\banachfunspace}[1]{\setL^{#1}}	
\safemath{\hilfunspace}{\banachfunspace{2}}	

\safemath{\SNR}{\textit{SNR}} 				
\safemath{\PAR}{\textit{PAR}} 				
\safemath{\No}{N_0}							
\safemath{\Es}{E_s}							
\safemath{\Eb}{E_b}							
\safemath{\EbNo}{\frac{\Eb}{\No}}
\safemath{\EsNo}{\frac{\Es}{\No}}

\DeclareMathOperator{\CHop}{\ensuremath{\opH}} 
\safemath{\tvir}{\rndh_{\CHop}}				
\safemath{\tvtf}{\rndl_{\CHop}}				
\safemath{\spf}{\rnds_{\CHop}}				
\safemath{\bff}{H_{\CHop}}					

\safemath{\ircf}{r_{h}}						
\safemath{\tftvcf}{r_{s}}					
\safemath{\tfcf}{r_{l}}						
\safemath{\bfcf}{r_{H}}						

\safemath{\tcorr}{c_h}						
\safemath{\scf}{c_{s}}						
\safemath{\tfcorr}{c_{l}}					
\safemath{\fcorr}{c_{H}}						

\safemath{\mi}{I}							
\safemath{\capacity}{C}						

\safemath{\normal}{\mathcal{N}}			
\safemath{\jpg}{\mathcal{CN}}			
\safemath{\mchain}{\leftrightarrow}		

\safemath{\dB}{\,\mathrm{dB}}
\safemath{\dBm}{\,\mathrm{dBm}}
\safemath{\Hz}{\,\mathrm{Hz}}
\safemath{\kHz}{\,\mathrm{kHz}}
\safemath{\MHz}{\,\mathrm{MHz}}
\safemath{\GHz}{\,\mathrm{GHz}}
\safemath{\s}{\,\mathrm{s}}
\safemath{\ms}{\,\mathrm{ms}}
\safemath{\mus}{\,\mathrm{\text{\textmu}s}}
\safemath{\ns}{\,\mathrm{ns}}
\safemath{\ps}{\,\mathrm{ps}}
\safemath{\meter}{\,\mathrm{m}}
\safemath{\mm}{\,\mathrm{mm}}
\safemath{\cm}{\,\mathrm{cm}}
\safemath{\m}{\,\mathrm{m}}
\safemath{\W}{\,\mathrm{W}}
\safemath{\mW}{\, \mathrm{mW}}
\safemath{\J}{\,\mathrm{J}}
\safemath{\K}{\,\mathrm{K}}
\safemath{\bit}{\,\mathrm{bit}}
\safemath{\nat}{\,\mathrm{nat}}

\safemath{\define}{\triangleq}			

\safemath{\equivalent}{\sim}
\safemath{\distas}{\sim}					
\safemath{\sdiff}{\Delta}				

\safemath{\reals}{\mathbb{R}}
\safemath{\positivereals}{\reals_{+}}
\safemath{\integers}{\mathbb{Z}}
\safemath{\posint}{\integers_{+}}
\safemath{\naturals}{\mathbb{N}}
\safemath{\posnaturals}{\naturals_{+}}
\safemath{\complexset}{\mathbb{C}}
\safemath{\rationals}{\mathbb{Q}}

\newcommand*{\fancyrefapplabelprefix}{app}		
\newcommand*{\fancyrefthmlabelprefix}{thm}		
\newcommand*{\fancyreflemlabelprefix}{lem}		
\newcommand*{\fancyrefcorlabelprefix}{cor}		
\newcommand*{\fancyrefdeflabelprefix}{def}		
\newcommand*{\fancyrefproplabelprefix}{prop}		
\newcommand*{\fancyrefexmpllabelprefix}{exmpl}
\newcommand*{\fancyrefalglabelprefix}{alg}		
\newcommand*{\fancyreftbllabelprefix}{tbl}		

\frefformat{vario}{\fancyrefseclabelprefix}{Sec.~#1}
\frefformat{vario}{\fancyrefthmlabelprefix}{Theorem~#1}
\frefformat{vario}{\fancyreftbllabelprefix}{Tbl.~#1}
\frefformat{vario}{\fancyreflemlabelprefix}{Lemma~#1}
\frefformat{vario}{\fancyrefcorlabelprefix}{Corollary~#1}
\frefformat{vario}{\fancyrefdeflabelprefix}{Definition~#1}
\frefformat{vario}{\fancyreffiglabelprefix}{Fig.~#1}
\frefformat{vario}{\fancyrefapplabelprefix}{Appendix~#1}
\frefformat{vario}{\fancyrefeqlabelprefix}{(#1)}
\frefformat{vario}{\fancyrefproplabelprefix}{Proposition~#1}
\frefformat{vario}{\fancyrefexmpllabelprefix}{Example~#1}
\frefformat{vario}{\fancyrefalglabelprefix}{Algorithm~#1}

 \newtheorem{thm}{Theorem}

 \newtheorem{lem}[thm]{Lemma}
 
\safemath{\dictab}{[\,\dicta\,\,\dictb\,]}

\safemath{\ysig}{\bmy}
\safemath{\ysighat}{\hat{\ysig}}
\safemath{\ysigdim}{M}
\safemath{\xsig}{\bmx}
\safemath{\xsigdim}{N}
\safemath{\nx}{n_x}
\safemath{\zsig}{\bmz}
\safemath{\zsigdim}{\ysigdim}
\safemath{\rsig}{\bmr}
\safemath{\Adict}{\bA}
\safemath{\Adicttilde}{\widetilde{\Adict}}
\safemath{\Adictdim}{\outputdim\times\xsigdim}
\safemath{\avec}{\bma}
\safemath{\avectilde}{\tilde{\avec}}
\safemath{\Bdict}{\bB}
\safemath{\Bdicttilde}{\widetilde{\Bdict}}
\safemath{\Cdict}{\bC}
\safemath{\cvec}{\bmc}
\safemath{\Ddict}{\bD}
\safemath{\Ddictdim}{\ysigdim\times\xsigdim}
\safemath{\dvec}{\bmd}
\safemath{\Ddicttilde}{\widetilde{\bD}}
\safemath{\Bonb}{\bB}
\safemath{\bvec}{\bmb}
\safemath{\Bonbdim}{\ysigdim\times\ysigdim}
\safemath{\noise}{\bmn}
\safemath{\noisedim}{\ysigim}
\safemath{\err}{\bme}
\safemath{\errdim}{\ysigdim}
\safemath{\errset}{\setE}
\safemath{\nerr}{n_e}
\safemath{\delop}{\bP_\errset}
\safemath{\delopc}{\bP_{{\errset}^c}}

\safemath{\cplxi}{\imath}
\safemath{\cplxj}{\jmath}

\safemath{\dict}{\matD}
\safemath{\inputdim}{N}		
\safemath{\outputdim}{M}		
\safemath{\sparsity}{S}	
\safemath{\inputdimA}{{N_a}}	
\safemath{\inputdimB}{{N_b}}	
\safemath{\elemA}{{n_a}}	
\safemath{\elemB}{{n_b}}	
\safemath{\resA}{\matR_a}	
\safemath{\resB}{\matR_b}	
\safemath{\subD}{\matS} 
\safemath{\subA}{\matS_a} 
\safemath{\subB}{\matS_b} 
\safemath{\dicta}{\matA} 	
\safemath{\dictb}{\matB} 	
\safemath{\hollowS}{H}
\safemath{\hollowA}{H_a}
\safemath{\hollowB}{H_b}
\safemath{\cross}{Z}
\safemath{\coh}{\mu_d}			
\safemath{\coha}{\mu_a}			
\safemath{\cohb}{\mu_b}			
\safemath{\mubs}{\nu}	
\safemath{\cohm}{\mu_m} 
\safemath{\dictset}{\setD}	
\safemath{\dictsetp}{\dictset(\coh,\coha,\cohb)}	
\safemath{\dictsetgen}{\dictset_\text{gen}}
\safemath{\dictsetgenp}{\dictsetgen(\coh)}
\safemath{\dictsetonb}{\dictset_\text{onb}}
\safemath{\dictsetonbp}{\dictsetonb(\coh)}

\safemath{\leftside}{U}
\safemath{\rightsideA}{R_a}
\safemath{\rightsideB}{R_b}

\safemath{\indexS}{\setI_S} 

\safemath{\na}{n_a}			
\safemath{\nb}{n_b}			
\safemath{\coeffa}{p_i}	
\safemath{\coeffb}{q_j}	
\safemath{\seta}{\setP}		
\safemath{\setb}{\setQ}     
\safemath{\setw}{\setW}	
\safemath{\setz}{\setZ}	
\safemath{\cola}{\veca}		
\safemath{\colb}{\vecb}		
\safemath{\cold}{\vecd}		
\safemath{\inputvec}{\vecx} 	
\safemath{\error}{\vece}	
\safemath{\noiseout}{\vecz} 	
\safemath{\inputvecel}{x}
\safemath{\inputveca}{\vecx_a}
\safemath{\inputvecb}{\vecx_b}
\safemath{\outputvec}{\vecy}	
\safemath{\lambdamin}{\lambda_{\mathrm{min}}}

\safemath{\elltwo}{\ell_2}
\safemath{\ellone}{\ell_1}
\safemath{\ellzero}{\ell_0}
\safemath{\ellinf}{\ell_\infty}
\safemath{\ellinftilde}{\ell_{\widetilde\infty}}
\safemath{\licard}{Z(\coh,\coha,\cohb)}
\safemath{\xsol}{\hat{x}}
\safemath{\xbord}{x_b}		
\safemath{\xstat}{x_s}		
\safemath{\xstatLone}{\tilde{x}_s}
\safemath{\order}{\mathcal{O}} 
\safemath{\scales}{\Theta} 
\safemath{\ones}{\mathbf{1}} 
\safemath{\zeroes}{\mathbf{0}} 
\safemath{\thlone}{\kappa(\coh,\cohb)} 
\safemath{\constoneA}{\delta} 
\safemath{\constoneB}{\epsilon} 
\safemath{\nlarge}{L}				   
\safemath{\sumlarge}{S_\nlarge}
\safemath{\maxlarger}{P_\nlarge}	   
\safemath{\Pzero}{\textrm{P0}}	
\safemath{\Pone}{\textrm{P1}}
\safemath{\vecfir}{\vecw}			 
\safemath{\vecsec}{\vecz}
\safemath{\elvecfir}{w}              
\safemath{\elvecsec}{z}				 
\safemath{\nlargefir}{n}
\safemath{\normout}{\gamma}
\safemath{\auxfun}{h}
\safemath{\supp}{\textrm{supp}}

\safemath{\indexa}{\ell}
\safemath{\indexb}{r}
\safemath{\indexc}{i}
\safemath{\indexd}{j}

\safemath{\project}{P}

\usepackage{framed}

\newtheorem{ass}{Assumption}
\newcommand*{\fancyrefasslabelprefix}{ass}
\frefformat{vario}{\fancyrefasslabelprefix}{Assumption~#1}

\IEEEoverridecommandlockouts

\begin{document}

\title{Mitigating SAR-ADC Non-Idealities in \\ Massive MU-MIMO Systems via Affine Models}

\author{J\'er\'emy Guichemerre and Christoph Studer\\[0.3cm]
\textit{Dept. of Information Technology and Electrical Engineering, ETH Zurich, Switzerland} \\ 
\textit{email: jeremyg@iis.ee.ethz.ch and studer@ethz.ch}
\thanks{The work of JG and CS was supported in part by an ETH Research Grant.}
\thanks{The authors would like to thank Seyed Hadi Mirfarshbafan and Darja Nonaca for their help with the system simulator.}}

\maketitle

\begin{abstract}
Low-resolution data converters can significantly reduce the power consumption and silicon area of all-digital massive multi-user (MU) multiple-input multiple-output (MIMO) basestations. However, the existing literature almost exclusively focuses on idealistic quantization models, neglecting the inherent non-idealities present in real-world analog-to-digital converter (ADC) implementations. To overcome this limitation, we propose two affine models, one based on Bussgang's decomposition and one that maximizes the signal-to-distortion ratio (SDR), both accounting for the most prominent non-idealities in successive approximation register (SAR) ADCs. Subsequently, we utilize these models to devise low-complexity methods that mitigate SAR-ADC non-idealities in massive MU-MIMO wireless systems.

\end{abstract}


\section{Introduction}

Millimeter-wave (mmWave) frequencies are targeted for fifth-generation (5G) and beyond-5G wireless communication systems as a solution to increasing the amount of available bandwidth~\cite{Rappaport13,Swindlehurst14}.
Operating at mmWave frequencies, however, comes at the cost of high path loss.
Massive multiple-input multiple-output (MIMO) is able to simultaneously mitigate the high path loss at mmWave frequencies and enables multiuser~(MU) communication capabilities~\cite{Lu14,Bjornsso17}.
Notwithstanding, the large number of antennas and analog front-ends (AFEs) required for all-digital massive MU-MIMO  results in excessively high power consumption and silicon area.

A popular remedy to reduce power and area is to employ low-resolution data converters~\cite{Jacobsson17,Mollen17}; the use of such data converters also relaxes the quality requirements of the AFEs.
The literature proposes a variety of methods that mitigate the adverse impact of coarse quantization~\cite{Oscar20,Mezghani10,Liang16} while keeping the advantages of all-digital beamformer architectures.
However, as pointed out recently in~\cite{Guichemerre23}, non-idealities in the data converters themselves can have a significant impact on the system performance, but are routinely ignored in the literature.

\subsection{Contributions}

We study the effects of hardware non-idealities in successive approximation register (SAR) ADCs~\cite{Suarez74}, which are commonly used for bandwidths ranging from tens of megahertz to low gigahertz at a medium resolution (5\,bit to 10\,bit)~\cite{adc_survey}.
We propose two affine signal decompositions for non-linear systems: (i) we generalize Bussgang's decomposition~\cite{Bussgang52} to an affine model and (ii) we propose the \emph{max-SDR decomposition} which maximizes the signal-to-distortion ratio (SDR).
In order to assess the properties of the Bussgang and max-SDR decompositions, we apply them to an ideal quantizer.
We then analyze the impact of capacitor mismatches in a single SAR ADC using both affine models and the analysis put forward in~\cite{Guichemerre23}.
We finally demonstrate the efficacy of affine corrections on quantized mmWave massive MU-MIMO systems.


\subsection{Notation}

We write matrices in bold uppercase, column vectors in bold lowercase, and sets in calligraphic letters. The Frobenius norm of a matrix $\bA$ is $\frobnorm{\bA}$.
The expectation and variance of a random variable~$X$ are $\Ex{X}$ and $\Var{X}$, respectively.
We write the derivative of a function~$f$ as~$f'$.
%


\section{Affine modeling of non-linear distortions} \label{sec:affine}

We now present an affine model of a non-linear function and review Bussgang's decomposition.
We then propose a new signal decomposition, the max-SDR decomposition, which maximizes the SDR.
Finally, we investigate the fundamental properties of these models for an ideal quantizer.

\subsection{Affine Model} \label{subsec:buss_model}

In order to model the impact of a non-linear function ${f\colon\opR\to\opR}$ applied to a real-valued zero-mean random variable~$X$ of variance ${\sigma^2_{\!X}>0}$, we propose the use of an affine model.
Specifically, the impact of $f$ is modeled as an offset $\eta\in\opR$, a gain $\beta\in\opR$, and an input-dependent distortion $D$:
\be \label{eq:linear_general}
f(X)=\beta X+ \eta +D.
\ee

In order to analyze this affine model, we must define several quantities. Before this, we make the following assumption:
\begin{ass} \label{ass:ass}
We require $\Ex{f(X)}$, $\Ex{f(X)^2}$, and $\Ex{Xf(X)}$\ to be finite. We will also require\footnote{Technically, the affine decompositions we present can be defined even if ${\Ex{D^2}=0}$, which means that $f$ is affine on the set of values the input $X$ can assume. This case is, however, of little practical relevance and prevents one from properly defining the SDR.} ${\Ex{D^2}>0}$.
\end{ass}

The first quantity we need is the power of the distortion $D$:
\begin{alignat}{2}
\Ex{D^2} 	&= \opE[( f(X)		-&& \beta X - \eta )^2]				\\
	 		&=	\Ex{f(X)^2} 		&&+ \beta^2 \sigma_{\!X}^2 - 2\beta\Ex{Xf(X)}	 \nonumber 	\\
			&					&&+ \eta^2 - 2\eta\Ex{f(X)}. \label{eq:d2_general}
\end{alignat}
The second quantity is the SDR, which we define as
\be \label{eq:sdr_general}
\mathrm{\textit{SDR}}	=	\frac{\Ex{(\beta X)^2}\!}{\Ex{D^2}} = \frac{\beta^2\sigma^2_{\!X}}{\Ex{D^2}}.
\ee

\subsection{Bussgang's Decomposition}

A well-known special case of the model in \fref{eq:linear_general} is Bussgang's decomposition~\cite{Bussgang52}.
Here, $\beta$ and $\eta$ are chosen so that the power of the distortion~$D$ is minimized.
A consequence of this decomposition, widely known as the orthogonality principle, is that the distortion~$D$ is uncorrelated to the input~$X$.

By taking the partial derivative of \fref{eq:d2_general} with respect to $\beta$ and~$\eta$, and by setting them to zero, we obtain the affine Bussgang gain~$\beta_b$ and offset~$\eta_b$:
\be\label{eq:buss_B}
\beta_\textnormal{b} = \frac{\Ex{Xf(X)}}{\sigma^2_{\!X}}\overset{(\conv)}{=}\Ex{f'(X)} \quad\textnormal{and}\quad \eta_\textnormal{b} = \Ex{f(X)}\!.
\ee
Here, $(\conv)$ holds if X is Gaussian (cf.~Stein's lemma \cite{Stein81}).
Note that the conventional Bussgang's decomposition does not include the offset~$\eta$ in \fref{eq:linear_general}, absorbing any residual affine component in the distortion~$D$.
However, the generalized affine Bussgang's decomposition coincides with the conventional \emph{linear} Bussgang's decomposition if ${\Ex{f(X)}=0}$.

Evaluating the SDR in \fref{eq:sdr_general} for the affine Bussgang's decomposition results in
\be \label{eq:buss_sdr}
\textit{SDR}_\textnormal{b} = \frac{\Ex{Xf(X)}^2}{\sigma_{\!X}^2\Var{f(X)}-\Ex{Xf(X)}^2}.
\ee

\subsection{Max-SDR Decomposition}

Bussgang's decomposition is widely used in the literature and is particularly useful if the distortion~$D$ should be uncorrelated with the input~$X$.
However, in many practical cases, it is more relevant to maximize the SDR.
We therefore propose a new signal decomposition that maximizes the SDR:

\newcommand{\footlemma}{\footnote{%
Note that the restriction ${\Ex{X f(X)}\neq 0}$ prevents the pathological case where the output is uncorrelated with the input. In this case, the Bussgang gain is zero and the gain that maximizes the SDR is also zero.
}}

\begin{lem} \label{lem:max_sdr}
Let \fref{ass:ass} hold. Furthermore, assume that ${\Ex{X f(X)}\neq 0}$. Then, the gain~$\beta_\textnormal{m}$ and offset~$\eta_\textnormal{m}$ that maximize the \textnormal{SDR} in \fref{eq:sdr_general} are\footlemma
\be\label{eq:max_B}
\beta_\textnormal{m} = \frac{\Var{f(X)}}{\Ex{X f(X)}}  \quad \textnormal{and} \quad \eta_\textnormal{m} = \Ex{f(X)}.
\ee
\end{lem}

\begin{proof}
We take the partial derivative of the SDR in $\eta$
\be \label{eq:max_dsdr_eta}
\frac{\partial\textit{SDR}}{\partial\eta} = -\frac{2\beta^2 \sigma^2_{\!X}}{\Ex{D^2}^2} \left(\eta-\Ex{f(X)}\right)\!,
\ee
and set it to zero to obain
\be \label{eq:max_eta}
\eta_\textnormal{m} = \Ex{f(X)}.
\ee

We insert \fref{eq:max_eta} into \fref{eq:d2_general}, and take the partial derivative with respect to $\beta$
\be \label{eq:max_dsdr_beta}
\frac{\partial\textit{SDR}}{\partial\beta} = \frac{2\beta \sigma^2_{\!X}}{\Ex{D^2}^2} \left(\Var{f(X)}-\beta\Ex{X f(X)}\right)\!,
\ee
and set it to zero to obtain
\be \label{eq:max_beta}
\beta_\textnormal{m} = \frac{\Var{f(X)}}{\Ex{Xf(X)}}.
\ee
Note that $\beta=0$ also nulls \fref{eq:max_dsdr_beta}, but one can verify that this solution does not maximize the SDR---it minimizes it!
\end{proof}

The following results link the max-SDR decomposition with Bussgang's decomposition.

\begin{lem} \label{lem:sdr_compare}
Assume the conditions in \fref{lem:max_sdr} hold. Then, the \textnormal{SDR} corresponding to the max-SDR decomposition satisfies
\be \label{eq:sdr_compare}
\textit{SDR}_\textnormal{m} = 1+\mathrm{\textit{SDR}_\textnormal{b}}.
\ee
\end{lem}

\begin{proof}
By inserting \fref{eq:max_eta} and \fref{eq:max_beta} into \fref{eq:sdr_general}, we obtain
\be
\textit{SDR}_\textnormal{m} = \frac{\sigma_{\!X}^2\Var{f(X)}}{\sigma_{\!X}^2\Var{f(X)}-\Ex{Xf(X)}^2}.
\ee
We compare this last expression to \fref{eq:buss_sdr} and obtain \fref{eq:sdr_compare}.
\end{proof}

Direct consequences of \fref{lem:sdr_compare} are that (i) $\textit{SDR}_\textnormal{m}$ is strictly greater than $\textit{SDR}_\textnormal{b}$ and (ii) $\textit{SDR}_\textnormal{m}$ is greater or equal to 1, which means that the power of the wanted signal cannot be lower than the power of the remaining distortion in the case of the max-SDR decomposition.
Again, this property of course only holds if all quantities are properly defined---we exclude the pathological cases mentioned earlier.

Another important consequence of \fref{lem:sdr_compare} is that the input power~$\sigma^2_{\!X}$ that maximizes the SDR at a given non-linearity~$f$ is the same both for Bussgang's decomposition and the {max-SDR} decomposition. This means that both decompositions are equivalent when it comes to input power-control schemes.
This equality of optimal input power can be proven by taking the partial derivative of \fref{eq:sdr_compare} with respect to $\sigma^2_{\!X}$.

\begin{lem} \label{lem:beta_compare}
Assume the conditions in \fref{lem:max_sdr} hold. Then, 
\be \label{eq:beta_compare}
\frac{\beta_\textnormal{m}}{\beta_\textnormal{b}} = 1 + \frac{1}{\textit{SDR}_\textnormal{b}}.
\ee
\end{lem}

\begin{proof}
We go back to the definitions of $\beta_\textnormal{b}$ and $\beta_\textnormal{m}$ from \fref{eq:buss_B} and \fref{eq:max_B} and write
\be
\frac{\beta_\textnormal{m}}{\beta_\textnormal{b}} = \frac{\sigma^2_{\!X}\Var{f(X)}}{\Ex{Xf(X)}^2}.
\ee
We then decompose the expression as follows:
\be
\frac{\beta_\textnormal{m}}{\beta_\textnormal{b}} = \frac{\sigma_{\!X}^2\Var{f(X)}-\Ex{Xf(X)}^2}{\Ex{Xf(X)}^2} + \frac{\Ex{Xf(X)}^2}{\Ex{Xf(X)}^2}.
\ee
We finally notice that the left side corresponds to the reciprocal of \fref{eq:buss_sdr}. Thus, we arrive at \fref{eq:beta_compare}.
\end{proof}

A direct consequence of \fref{lem:beta_compare} is that $\beta_\textnormal{m}$ is not smaller than $\beta_\textnormal{b}$ and strictly larger if the SDR is finite.
\fref{lem:beta_compare} also implies that the Bussgang and max-SDR decomposition are virtually the same for large values of SDR: $\beta_\textnormal{b}$ tends to~$\beta_\textnormal{m}$ when $\textit{SDR}_\textnormal{b}$ tends to infinity (which is the same as $\textit{SDR}_\textnormal{m}$ tending to infinity).
However, if $f$ has a strong effect on the SDR and makes it lower or close to 1, then the decompositions will be very different as the $1/\textit{SDR}_\textnormal{b}$ term will become larger.

\subsection{Application to an Ideal Quantizer}

\begin{figure*}[tp]
\centering
	\subfigure[SDR, correlation, and distortion power.]{\includegraphics[width=0.32\linewidth]{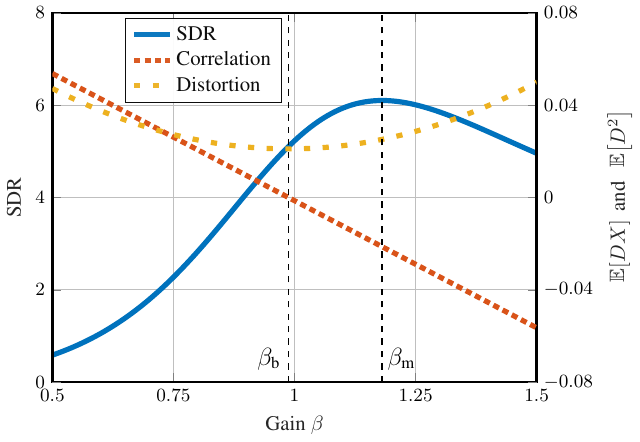}\label{fig:ideal_adc_sdr}}
\hfill
	\subfigure[Effective resolution (EFR).]{\includegraphics[width=0.28\linewidth]{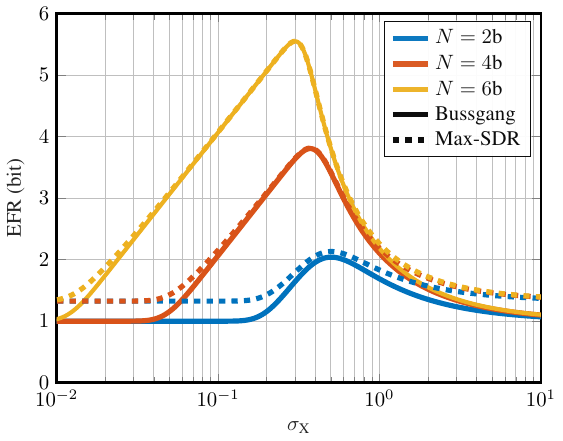}\label{fig:ideal_adc_efr}}
\hfill
	\subfigure[Model gains $\beta_\textnormal{b}$ and $\beta_\textnormal{m}$.]{\includegraphics[width=0.3\linewidth]{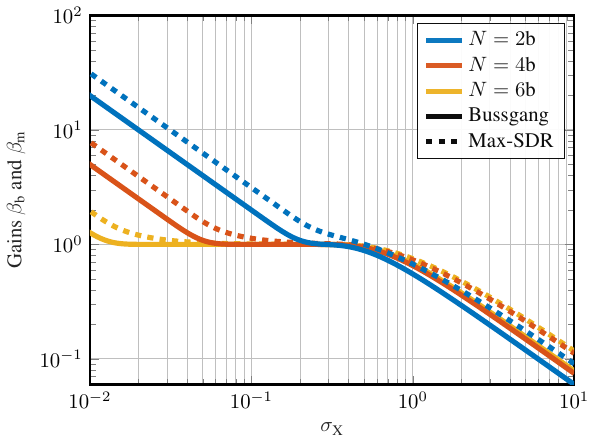}\label{fig:ideal_adc_beta}}
\vspace{-0.15cm}  
\caption{(a) SDR, distortion power, and input to distortion correlation at a fixed $\sigma^2_{\!X}$ for a 2-bit uniform symmetric mid-rise quantizer. The chosen input variance~$\sigma^2_{\!X}$ enables one to reach the maximum achievable SDR for such a quantizer. (b) EFR of $N$-bit uniform symmetric mid-rise quantizers with Bussgang's and max-SDR decompositions as a function of the input variance $\sigma^2_{\!X}$. (c) Bussgang's gain $\beta_\textnormal{b}$ and max-SDR gain $\beta_\textnormal{m}$ as a function of the input variance $\sigma^2_{\!X}$.}
\end{figure*}

We now illustrate the difference between the Bussgang and max-SDR decompositions with the example of an ideal 2\,bit uniform symmetric mid-rise quantizer with input range $[-1,1]$, where values outside this range are clipped.
We consider a zero-mean Gaussian input~$X$ of variance $\sigma^2_{\!X}$.
Note that $\eta=0$ for both decompositions as $f$ is symmetric.

We first set the input power~$\sigma^2_{\!X}$ so that the highest possible SDR for the Bussgang and max-SDR decompositions is achieved.
From \fref{lem:sdr_compare} follows that this optimum is the same for both decompositions.
\fref{fig:ideal_adc_sdr} shows the SDR at this fixed input power for a varying gain~$\beta$.
We see that the max-SDR decomposition indeed maximizes the SDR whereas Bussgang's decomposition does not.
We also verify that Bussgang's decomposition minimizes the distortion power while keeping the input~$X$ and the distortion~$D$ uncorrelated.
Note that \fref{fig:ideal_adc_sdr} is generated by evaluating \fref{eq:sdr_general}, \fref{eq:d2_general}, and $\Ex{DX}$ using the analytical results derived in~\cite{Guichemerre23}.

In order to quantify the resolution of a quantizer, we express its effective resolution (EFR) as~\cite{Guichemerre23}
\be \label{eq:EFROMG}
\textit{EFR} \define \frac{\textit{SDR}_\text{dB}\!+\!10\log_{10}(2(\pi \!-\!2))}{20\log_{10}(2)}  
					\approx \frac{\textit{SDR}_\text{dB}\!+\!3.59}{6.02},
\ee
with $\textit{SDR}_\text{dB}=10\log_{10}(\textit{SDR})$.
The EFR is an alternative to the effective number of bits (ENOB)~\cite{Walden99}. The ENOB is generally used to assess the resolution of an ADC, but it is measured with a full-scale sinusoidal input.
In contrast, the EFR, as defined in \fref{eq:EFROMG}, is based on Bussgang's decomposition and is measured with a zero-mean Gaussian input.
The EFR  increases by one bit for every doubling of the ratio between the wanted signal amplitude and the amplitude of the distortion error, and sets the EFR of an ideal 1\,bit quantizer to 1\,bit.
\fref{fig:ideal_adc_efr} shows the EFR of ideal quantizers as a function of the input power~$\sigma^2_{\!X}$ when their output are interpreted with either Bussgang's decomposition or the proposed max-SDR decomposition.

As expected from \fref{lem:sdr_compare}, both decompositions lead to similar SDR in regimes where the SDR is high.
For the same input powers, \fref{fig:ideal_adc_beta} illsutrates \fref{lem:beta_compare} by showing that the gains $\beta_\textnormal{b}$ and $\beta_\textnormal{m}$ are  close to each other when the SDR (and, hence, also the EFR) is high.

To conclude the discussion of this section, we have proposed a new decomposition, the max-SDR decomposition, that is preferable over Bussgang's decomposition in cases where uncorrelatedness between the input and the residual distortion is not required.
We have also generalized both decompositions to affine models that take a possible offset induced by the non-linearity into account.  
We have shown that optimizing the input power for one of the decomposition also optimizes it for the other decomposition.
Finally, we have proved that the Bussgang and max-SDR decompositions are asymptotically equivalent when the SDR tends to infinity.


\section{SAR-ADC Mismatch: Model and compensation} \label{sec:SARmismatch}

\begin{figure}[tp] 
	\centering
	\includegraphics[width=0.32\textwidth]{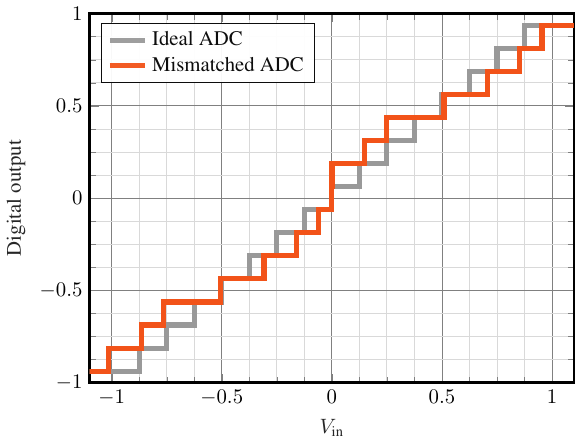}
	\vspace{-0.3cm}    
	\caption{Transfer function of a $4$\,b SAR ADC with mismatches on all capacitors.}
	\label{fig:miss_tf}
\end{figure}

\subsection{Mismatches in SAR ADCs}

A successive approximation register (SAR) ADC implements a binary search.
Modern implementations operate in differential mode, meaning that the input to be quantized is the voltage difference between a positive and a negative node.
\fref{fig:SAR_blocks} shows a high-level block diagram of a generic SAR ADC with $V_\text{in}= V_\text{in}^+ - V_\text{in}^-$.
The most significant bit (MSB) is computed directly after sampling.
Next, we need to recenter the chosen quadrant of the binary search.
To this end, one of the MSB capacitors $C_\text{MSB}$ of the capacitive digital-to-analog converter (CDAC) switches (which one depends on the value of the MSB).
Another comparison is then performed to extract the next bit, after which another quadrant correction is carried out.
This quadrant correction is performed by one of the two $C_\text{MSB}/2$ capacitors (half of the value of the first ones because the required correction is half as large).
The cycle of comparison and shift continues as many times as needed to extract the desired number of bits.
See \cite{Chan11} for more details on the operating principles of SAR ADCs.

\begin{figure}[tp] 
	\centering
	\includegraphics[width=0.49\textwidth]{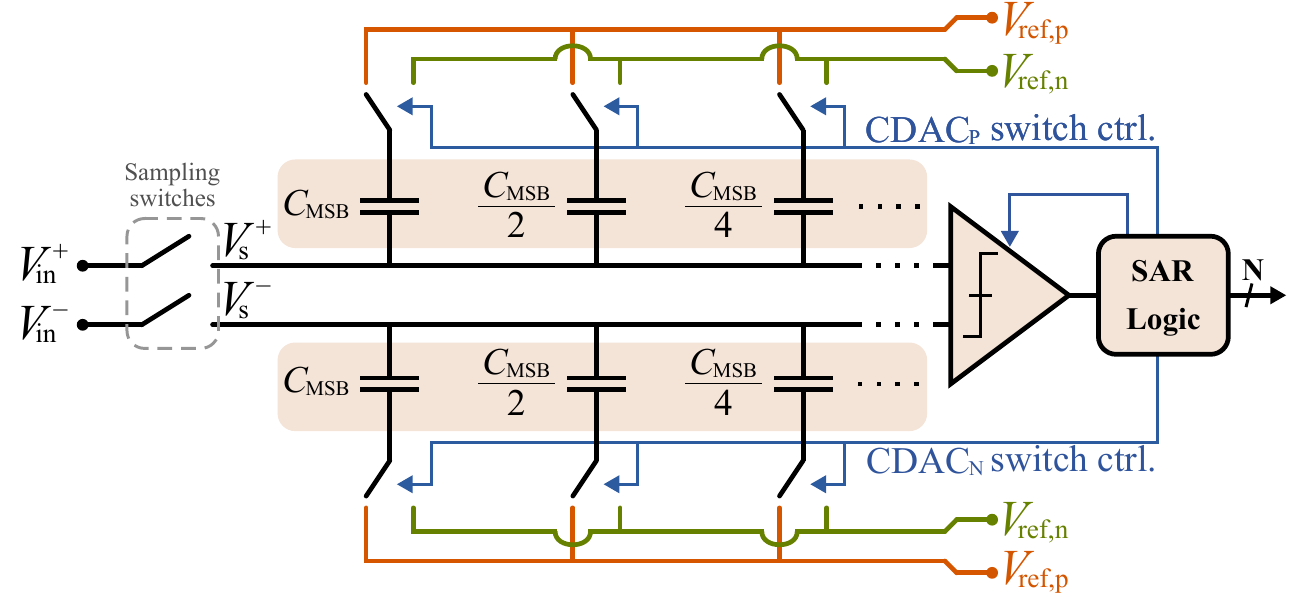}
	\vspace{-0.6cm}    
	\caption{High-level block diagram of a SAR ADC~\cite{Guichemerre23}.}
	\label{fig:SAR_blocks}
\end{figure}

\begin{figure*}[tp]
\centering
	\subfigure[CDF of the EFR with high mismatch]{\includegraphics[width=0.3\linewidth]{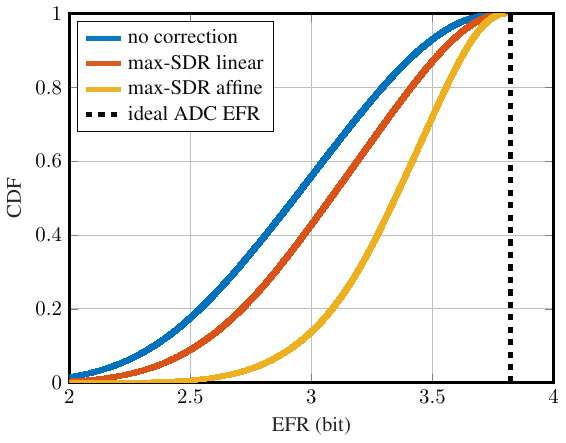}\label{fig:miss_cdf_high}}
\hfill
	\subfigure[CDF of the EFR with low mismatch]{\includegraphics[width=0.3\linewidth]{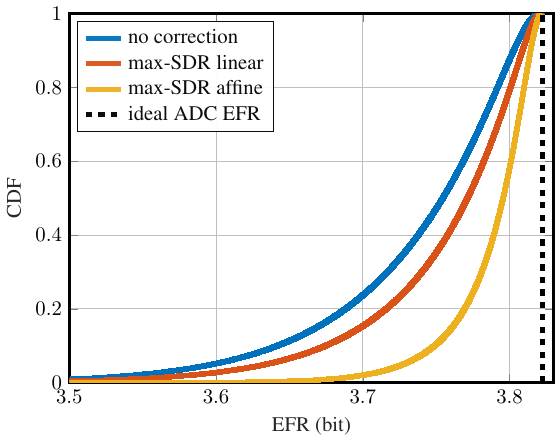}\label{fig:miss_cdf_low}}
\hfill
	\subfigure[Same as (b), log scale]{\includegraphics[width=0.3\linewidth]{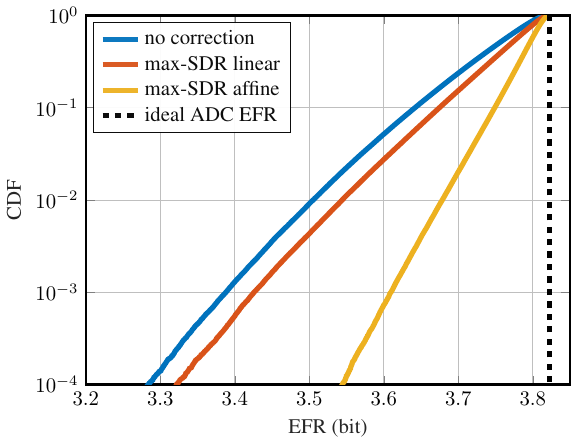}\label{fig:miss_cdf_low_log}}
\vspace{-0.15cm}  
\caption{CDF of the EFR of a 4b SAR ADC with mismatches on all capacitors with large mismatch ($\sigma_m=\Delta/2$) and small mismatch ($\sigma_m=\Delta/10$).}
\end{figure*}

Manufacturing processes used to fabricate chips are not ideal, which will affect the properties of the circuit's components in the design~\cite{Plassche94}.
In particular, capacitors will suffer from mismatches: the average value of a capacitor across many chips will be the intended one, but the capacitor value in one specific fabricated chip is the realization of a random process, typically assumed to be Gaussian with variance $\sigma_{\text{C}}^2$~\cite{Pelgrom89}.

In a SAR ADC, such mismatches lead to a non-ideal binary search as the quadrant shifts deviate from the ideal ones.
Due to such mismatches, the improper MSB capacitor value will induce "jumps" or "saddle points" (depending on the sign of the mismatch) in the input-output transfer function (TF) of the ADC (see \fref{fig:miss_tf} for an example of such a TF).
Considering an ADC with range $[-1, 1]$, this non-linearity happens at $0$ (the center), $-1$, and $1$ (the extremes of the range).
Mismatches of the next set of capacitors  will then induce similar effects, but at $-1$, $-0.5$, $0$, $0.5$, and $1$.
In general, capacitors associated to the $n\,$th bit will induce $2^{n+1}-1$ "jumps" or "saddle points" at the boundaries of the already computed $n$-bit code. 
Because of that, unlike the TF of an ideal quantizer, the TF of a SAR ADC with capacitor mismatches does not have a simple analytical form.
The issue is that mismatches of large capacitors change the influence of the mismatches of the smaller capacitors used later in the binary search.
In cases where an early error already leads to wrong quadrant corrections, many quantization steps become inaccessible (like the first positive output code in \fref{fig:miss_tf}), therefore nulling the effect of the mismatches of the subsequent less significant bits.
Due to this effect, one cannot simply add up the impact of the mismatches of each capacitor.
Nonetheless, the necessary quantities (cf. \fref{ass:ass}) can still be obtained easily using Monte--Carlo methods.

\subsection{MSB Mismatch Model}

Capacitor mismatches can lead to an offset as their effect on the ADC's TF is not symmetric.
We therefore extend the derivation from~\cite{Guichemerre23}, where we only model the effect of clipping and MSB mismatches, to include the effect of such an offset.
To perform an analysis, we assume real-valued zero-mean Gaussian inputs $X$ with variance $\sigma^2_{\!X}$.
In order to generalize the linear model from~\cite{Guichemerre23} to the affine model in \fref{eq:buss_B} or \fref{eq:max_B}, we only need the additional quantity $\Ex{f(X)}$.

As for the case of linear modeling, $m_1$ is the width of the jump or saddle point induced by the MSB mismatch in $\textnormal{CDAC}_\textnormal{P}$,~$m_2$, is the one for $\textnormal{CDAC}_\textnormal{N}$ (see~\cite{Guichemerre23} for more details).
We first compute the influence of a mismatch to positive inputs only by defining $g(m)$ as
\be
g(m) = \int_{0}^{+\infty}\!\!\!\! f(t) \phi_{\sigma_{\!X}}(t)\,\textnormal{d}t,
\ee
where $\phi_{\sigma_{\!X}}$ is the probability density function (PDF) of a real-valued zero-mean Gaussian with variance $\sigma_{\!X}^2$.

Evaluating the integral in the case of positive and negative mismatch leads to
\begin{alignat}{2}
&g_{m \ge 0}(m) = \sigma_{\!X} 	&&\left( \phi\!\left(\frac{m}{\sigma_{\!X}}\right)\! - \phi\!\left(\frac{1+m}{\sigma_{\!X}}\right)\!\right)\! \nonumber\\
&							&&+ (1+m) Q\!\left( \frac{1+m}{\sigma_{\!X}}\right)\! - m Q\!\left( \frac{m}{\sigma_{\!X}}\right) 					\\
&g_{m < 0}(m) = \sigma_{\!X} 	&&\left( \phi\!\left( 0 \right)\! - \phi\!\left(\frac{1+m}{\sigma_{\!X}}\right)\!\right)\! 					\nonumber\\
&							&&+ (1+m) Q\!\left( \frac{1+m}{\sigma_{\!X}}\right)\! - \frac{m}{2},
\end{alignat}
where $\phi(\cdot)$ is the distribution of a standard normal random variable and $Q(x)= \int_{x}^\infty\! \phi(t) \,\textnormal{d}t$ the Q-function.

The offset $\eta$ for both the Bussgang's and max-SDR decompositions then corresponds to
\be \label{eq:msb_affine}
\eta = \Ex{f(X)} = g(m_1)-g(m_2).
\ee

The result of \fref{eq:msb_affine} can be used to extend the linear model presented in~\cite{Guichemerre23} to also model possible offsets and to achieve more precise predictions of the effect of mismatches of SAR ADCs while still only modeling MSB mismatches and clipping.

\subsection{Affine Compensation} \label{sec:comp}

We now use our affine models to mitigate the impact of capacitor mismatches in SAR ADCs. 
To this end, we compare three cases: (i) an uncompensated baseline, (ii) linear compensation, and (iii) affine compensation.
The uncompensated baseline (i) corresponds to a mismatched SAR ADC in which the gain~$\beta$ is fixed to 1 and the offset $\eta$ to 0.
This models the situation where the influence of the quantizer and its mismatches are completely ignored
The linear compensation (ii) uses the linear version of the max-SDR decomposition by setting
\be \label{eq:maxsdr_lin}
\beta_\textnormal{m,lin}=\frac{\Ex{f(X)^2}}{\Ex{Xf(X)}} \quad \textnormal{and} \quad  \eta_\textnormal{m,lin}=0.
\ee
The effects of any offset is therefore absorbed in the distortion~$D$ and gain $\beta_\textnormal{m,lin}$.
The estimate of the input resulting from this decomposition is then $f(X)$ divided by the gain of \fref{eq:maxsdr_lin}.
Finally, we perform an affine compensation (iii) by applying the max-SDR decomposition of \fref{eq:max_B}.
This time, we  subtract the offset $\eta_\textnormal{m}$ from $f(X)$, and then divide by the gain $\beta_\textnormal{m}$.

For each of the three scenarios, we will report the EFR.
However, we need to show the distribution of the EFR and not only its average.
Indeed, there will be no averaging effect as the mismatches are fixed after a chip is manufactured.
We will therefore report the cumulative distribution function (CDF).


\fref{fig:miss_cdf_high} shows the CDF of the EFR over two millions mismatch realizations of a 4b SAR ADC with mismatches on all capacitors of its CDACs.
Note that we estimate the gain and offsets associated with the three compensation methods with a Monte--Carlo run of one million inputs, and this for \textit{each} mismatch realization.
The input is a zero-mean Gaussian and its variance fixed to the value maximizing the EFR of an ideal 4b ADC.
The mismatch standard deviation $\sigma_m$ of the MSB capacitors is set to a high value of $\Delta/2$, where $\Delta=0.125$ corresponds to  a least-significant bit (LSB) step.

Taking the $10$\%-quantile as a measure, meaning that $10$\% of the manufactured ADCs will have an EFR worse than what we report here, we obtain an EFR of $2.35$\,b without correction, $2.52$\,b with linear correction, and $2.91$\,b with affine correction. 
We notice a strong influence of the compensation on the achievable resolution, and additionally that considering the offset significantly lowers the distortion power.

\fref{fig:miss_cdf_low} shows the same setup as \fref{fig:miss_cdf_high}, but with a lower mismatch standard deviation of $\sigma_m=\Delta/10$. The $10$\%-quantiles of the EFR are this time $3.59$\,b, $3.62$\,b, and $3.69$\,b.
We see that compensation for the offset on top of the gain compensation has, once again, a larger impact than between linear compensation and no compensation at all.
However, because the EFR is already close to its maximum (which is the EFR of an ideal ADC), we see less of an absolute influence in terms of effective resolution.

The significance of the offset compensation even with low mismatch is shown in \fref{fig:miss_cdf_low_log}.
We show the same CDF as in \fref{fig:miss_cdf_low}, but on a  logarithmic scale.
We notice that the slope is significantly higher in the case of an affine compensation compared to that of the other methods.
This means that the EFR difference resulting from the studied methods will be accentuated when lower quantiles are chosen.
Indeed, our choice of $10$\%-quantile would not fit any industrial grades: $10$\% of the manufactured chips would not meet the target specification, resulting in a disastrous production yield.
If we consider a $0.1$\%-quantile instead, we then see a difference of approximately $0.2$\,b (corresponding to $1.2$\,dB better SDR) between linear and affine compensation methods.

In short, it is paramount to compensate not only for the gain, but also for the offset as the latter has a large impact on the resolution. 
When taking more realistic assumptions on the target production yield (e.g., the $0.1$\%-quantile mentioned above), the offset plays an even more critical role in the distortion.


\section{Impact of SAR ADC Mismatch on Quantized Massive MU-MIMO Systems}\label{sec:mimo_sim}

We now apply an affine correction to a mmWave massive MU-MIMO uplink system with mismatched ADCs based on the max-SDR affine decomposition introduced in \fref{sec:affine}.

\subsection{System Model}

We simulate $U=16$ single-antenna user equipments (UEs) transmitting 16-QAM data to a $B=64$ antennas basestation (BS) arranged in a uniform linear array.
The UEs utilize a non-ideal power control, staying within $\pm 3$\,dB from the average power $\Es$.  
We use a QuaDRiGa mmMagic channel model~\cite{jaeckel2019quadriga} with $1^\circ$ minimum angular separation between the UEs.
The UEs first transmit orthogonal pilots from which the BS performs conventional least-squares channel estimation.
The BS then estimates the transmitted data symbols with a linear minimum mean-square error (LMMSE) data detector.
We adopt the following discrete-time frequency-flat system~model:
\be \label{eq:mimo_inout}
	\bmy = (g \circ f)(\bmz) \quad \text{with} \quad \bmz= \bH \bms + \bmn .
\ee
Here, $\bmy\in\complexset^B$ is the BS receive vector resulting from the composition of the functions $f$ and $g$ applied to the ideal receive vector $\bmz\in\complexset^B$.
The function $f$ is the TF of the ADCs, including quantization, mismatch, and clipping, applied element-wise on the real and imaginary part separately, whereas the function $g$ is the chosen correction function, which is the identity if no mismatch mitigation is applied and an affine function if affine mitigation is used. The SAR-ADC mismatches applied by the function $f$ are characterized by the standard deviation $\sigma_m$ of the impact of the MSB mismatch, but mismatches are applied to all capacitors (see~\cite{Guichemerre23} for more details).
$\bH\in\complexset^{B\times U}$ is the effective channel matrix including power control, $\bms\in\setS^U$ contains the data symbols transmitted by the UEs, and $\bmn\in\complexset^B$ contains realizations of i.i.d.\ circularly-symmetric complex Gaussians with variance $\No$ modeling thermal noise.
We define the average receive signal-to-noise ratio (SNR) as follows:
\be	\label{eq:snr_def}
	\textit{SNR} \define \frac{\frobnorm{\bH}^2E_s}{BN_0}.
\ee

If affine compensation is used, then the offset and gain as defined by the max-SDR decomposition are respectively subtracted and divided per ADC element.
For each mismatch realization, the parameters of~\fref{eq:max_B} are estimated by applying Gaussian inputs to the ADC array.

\begin{figure}[tp] 
	\centering
	\includegraphics[width=0.32\textwidth]{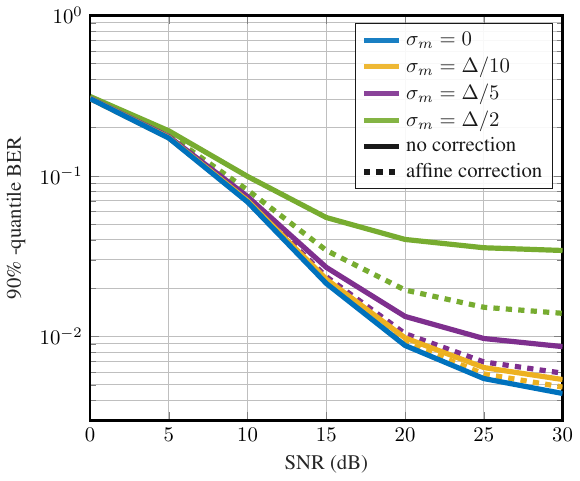}
	\vspace{-0.3cm}    
	\caption{Massive MU-MIMO simulation with mismatched 4b SAR ADC.}
	\label{fig:mimo_fig}
\end{figure}

\subsection{Simulation Results}\label{sec:sim_results}

We simulate the system described above with $4$\,bit mismatched SAR ADCs, with and without affine mismatch compensation, for different MSB mismatch standard deviation~$\sigma_m$.
We express values taken by $\sigma_m$ as a fraction of the size of an LSB-step $\Delta$.
Note that $\Delta=0.125$ in the case of a $4$\,bit ADC with input range $[-1,1]$.

As pointed out in \fref{sec:comp}, we need to report a quantile of the BER because mismatches are fixed once a chip is manufactured.
We choose the 90\%-quantile of the uncoded bit error rate (BER) as a metric to assess system performance.
We report the 90\%-quantile because a lower BER is better than a high one.
The chosen measure therefore corresponds to the BER requirement that 90\% of the produced chips will meet.

From \fref{fig:mimo_fig}, we see that an affine compensation greatly improves the BER when the mismatches would otherwise drastically degrade system performance (i.e.,~when $\sigma_m$ is large).
For example, at $\sigma_m=\Delta/2$, the uncompensated MIMO system has a BER floor greater of $3.5\times 10^{-2}$, but an affine compensation allows to half this floor to $1.4\times 10^{-3}$ where forward error-correction can be effective.
Note that we observe similar behavior independent of the system size:
by keeping a constant ratio between $B$ and $U$, the BER-quantile results stay virtually the same, whereas a different ratio changes the overall shapes of the curves but not the impact of affine compensation.


\section{Conclusions}

We have analyzed the impact of capacitor mismatches in widely-used successive approximation (SAR) ADCs.
To this end, we have proposed a novel affine signal decomposition for non-linear functions, the max-SDR decomposition, that maximizes the SDR---we have also generalized Bussgang's decomposition to an affine model.
After investigating the commonalities and differences between the max-SDR and Bussgang decompositions, we have applied them to the example of an ideal quantizer.
We then have extended the work of~\cite{Guichemerre23} on SAR-ADC mismatch modeling and showed the superiority of affine modeling over linear modeling as well as the efficiency of affine-model-based mitigation of SAR-ADC capacitor mismatches.
Finally, we have simulated a mmWave massive MU-MIMO uplink system and demonstrated that an affine correction can greatly improve the error rates.
Our results have shown that the design requirements of the analog front-ends can be relaxed by applying suitable compensation schemes, enabling savings in both power and silicon area.

\balance


\bibliographystyle{IEEEtran}
\bibliography{}

\begin{thebibliography}{10}
\providecommand{\url}[1]{#1}
\csname url@samestyle\endcsname
\providecommand{\newblock}{\relax}
\providecommand{\bibinfo}[2]{#2}
\providecommand{\BIBentrySTDinterwordspacing}{\spaceskip=0pt\relax}
\providecommand{\BIBentryALTinterwordstretchfactor}{4}
\providecommand{\BIBentryALTinterwordspacing}{\spaceskip=\fontdimen2\font plus
\BIBentryALTinterwordstretchfactor\fontdimen3\font minus
  \fontdimen4\font\relax}
\providecommand{\BIBforeignlanguage}[2]{{%
\expandafter\ifx\csname l@#1\endcsname\relax
\typeout{** WARNING: IEEEtran.bst: No hyphenation pattern has been}%
\typeout{** loaded for the language `#1'. Using the pattern for}%
\typeout{** the default language instead.}%
\else
\language=\csname l@#1\endcsname
\fi
#2}}
\providecommand{\BIBdecl}{\relax}
\BIBdecl

\bibitem{Rappaport13}
T.~S. Rappaport, S.~Sun, R.~Mayzus, H.~Zhao, Y.~Azar, K.~Wang, G.~N. Wong,
  J.~K. Schulz, M.~Samimi, and F.~Gutierrez, ``Millimeter wave mobile
  communications for {5G} cellular: It will work!'' \emph{IEEE Access}, vol.~1,
  pp. 335--349, May 2013.

\bibitem{Swindlehurst14}
A.~L. Swindlehurst, E.~Ayanoglu, P.~Heydari, and F.~Capolino, ``Millimeter-wave
  massive {MIMO}: the next wireless revolution?'' \emph{{IEEE} Commun. Mag.},
  vol.~52, no.~9, pp. 56--62, Sep. 2014.

\bibitem{Lu14}
L.~Lu, G.~Y. Li, A.~L. Swindlehurst, A.~Ashikhmin, and R.~Zhang, ``An overview
  of massive {MIMO}: Benefits and challenges,'' \emph{{IEEE} J. Sel. Topics
  Signal Process.}, vol.~8, no.~5, pp. 742--758, Apr. 2014.

\bibitem{Bjornsso17}
E.~Björnson, J.~Hoydis, and L.~Sanguinetti, \emph{Massive {MIMO} Networks:
  Spectral, Energy, and Hardware Efficiency}.\hskip 1em plus 0.5em minus
  0.4em\relax Now Publishers, 2017, vol.~11, no. 3-4, ser. Fundations and
  Trends in Signal Processing.

\bibitem{Jacobsson17}
S.~Jacobsson, G.~Durisi, M.~Coldrey, U.~Gustavsson, and C.~Studer, ``Throughput
  analysis of massive {MIMO} uplink with low-resolution {ADCs},'' \emph{{IEEE}
  Trans. Wireless Commun.}, vol.~16, no.~6, pp. 4038--4051, Apr. 2017.

\bibitem{Mollen17}
C.~Mollén, J.~Choi, E.~G. Larsson, and R.~W. Heath, ``Uplink performance of
  wideband massive {MIMO} with one-bit {ADCs},'' \emph{{IEEE} Trans. Wireless
  Commun.}, vol.~16, no.~1, pp. 87--100, Oct. 2017.

\bibitem{Oscar20}
O.~Casta\~{n}eda, S.~Jacobsson, G.~Durisi, T.~Goldstein, and C.~Studer,
  ``Finite-alphabet {MMSE} equalization for all-digital massive {MU}-{MIMO}
  {mmWave} communication,'' \emph{{IEEE} J. Sel. Areas Commun.}, vol.~38,
  no.~9, pp. 2128--2141, Jun. 2020.

\bibitem{Mezghani10}
A.~Mezghani, F.~Antreich, and J.~A. Nossek, ``Multiple parameter estimation
  with quantized channel output,'' in \emph{Proc. Int. ITG Workshop Smart
  Antennas (WSA)}, Feb. 2010, pp. 143--150.

\bibitem{Liang16}
N.~Liang and W.~Zhang, ``Mixed-{ADC} massive {MIMO},'' \emph{{IEEE} J. Sel.
  Areas Commun.}, vol.~34, no.~4, pp. 983--997, Mar. 2016.

\bibitem{Guichemerre23}
J.~Guichemerre and C.~Studer, ``The impact of {SAR}-{ADC} mismatch on quantized
  massive {MU-MIMO} systems,'' in \emph{Proc. 57th Asilomar Conf. Signals,
  Syst., Comput.}, Oct. 2023.

\bibitem{Suarez74}
R.~Suarez, P.~Gray, and D.~Hodges, ``{An all-MOS charge-redistribution A/D
  conversion technique},'' in \emph{IEEE Int. Solid-State Circuits Conf.
  (ISSCC)}, vol. XVII, Feb. 1974, pp. 194--195.

\bibitem{adc_survey}
{B.~Murmann}, ``{ADC Performance Survey 1997-2023},'' [Online]. Available:
  \url{https://github.com/bmurmann/ADC-survey}.

\bibitem{Bussgang52}
J.~J. Bussgang, ``Crosscorrelation functions of amplitude-distorted {Gaussian}
  signals,'' Massachusetts Inst. Technol., Cambridge, MA, USA, Tech. Rep., Mar.
  1952.

\bibitem{Stein81}
\BIBentryALTinterwordspacing
C.~M. Stein, ``Estimation of the mean of a multivariate normal distribution,''
  \emph{Ann. Statist.}, vol.~9, no.~6, pp. 1135--1151, Nov. 1981. [Online].
  Available: \url{http://www.jstor.org/stable/2240405}
\BIBentrySTDinterwordspacing

\bibitem{Walden99}
R.~Walden, ``Analog-to-digital converter survey and analysis,'' \emph{{IEEE} J.
  Sel. Areas Commun.}, vol.~17, no.~4, pp. 539--550, Apr. 1999.

\bibitem{Chan11}
T.~Chan~Carusone, D.~Johns, and K.~Martin, \emph{Analog Integrated Circuit
  Design, 2nd Edition}.\hskip 1em plus 0.5em minus 0.4em\relax Wiley, 2011.

\bibitem{Plassche94}
R.~J. van~der Plassche, \emph{Integrated Analog-To-Digital and
  Digital-To-Analog Converters}.\hskip 1em plus 0.5em minus 0.4em\relax Kluwer
  Academic Publischers, 1994.

\bibitem{Pelgrom89}
M.~Pelgrom, A.~Duinmaijer, and A.~Welbers, ``Matching properties of {MOS}
  transistors,'' \emph{{IEEE} J. Solid-State Circuits}, vol.~24, no.~5, pp.
  1433--1439, Oct. 1989.

\bibitem{jaeckel2019quadriga}
S.~Jaeckel, L.~Raschkowski, K.~Boerner, L.~Thiele, F.~Burkhardt, and
  E.~Eberlein, ``{QuaDRiGa} - quasi deterministic radio channel generator, user
  manual and documentation,'' Fraunhofer Heinrich Hertz Inst., Tech. Rep.
  v2.2.0, Jun. 2019.

\end{thebibliography}

\balance

\end{document}